\documentclass[a4paper]{jgaa-art}

\usepackage{graphicx}
\usepackage{amsmath}
\usepackage{amssymb}
\usepackage[utf8]{inputenc} %
\usepackage[margin=0pt,font=small,labelfont=bf]{caption}
\usepackage{subcaption} %
\usepackage{enumerate} %
\usepackage{aliascnt}
\usepackage{xspace}
\usepackage{graphicx}
\usepackage{tikz} %
\usetikzlibrary{automata,arrows,backgrounds,decorations.markings,decorations.pathmorphing,decorations.pathreplacing,fit,positioning,shadows,shapes,shapes.geometric,calc} 
\usepackage{todonotes}
\usepackage{cancel}
\usepackage{cleveref}
\usepackage{bold-extra}

\usepackage{lineno}
\usepackage{cite} %

\newcommand{\nodesVX}{
	\node[draw,circle,minimum size=0.45cm, inner sep=0ex] 
	(v) at (0, 0) {\small $v$};
	\node[draw,circle,minimum size=0.45cm, inner sep=0ex] 
	(w1) at (2, 2) {\small $x$};
}
\newcommand{\nodesVXcol}{
	\node[draw,circle,minimum size=0.45cm, inner sep=0ex,nodeset one] 
	(v) at ( 0, 0) {\small $v$};
	\node[draw,circle,minimum size=0.45cm, inner sep=0ex,nodeset one] 
	(w1) at (2, 2) {\small $x$};
}
\newcommand{\nodesVXcolV}{
	\node[draw,circle,minimum size=0.45cm, inner sep=0ex,nodeset one] 
	(v) at ( 0, 0) {\small $v$};
	\node[draw,circle,minimum size=0.45cm, inner sep=0ex] 
	(w1) at (2, 2) {\small $x$};
}
\newcommand{\nodesVXcolX}{
	\node[draw,circle,minimum size=0.45cm, inner sep=0ex,nodeset one] 
	(v) at ( 0, 0) {\small $v$};
	\node[draw,circle,minimum size=0.45cm, inner sep=0ex,nodeset two] 
	(w1) at (2, 2) {\small $x$};
}
\newcommand{\nodesWY}{
	\node[draw,circle,minimum size=0.45cm, inner sep=0ex] 
	(y) at (4, 0) {\small $w$};
	\node[draw,circle,minimum size=0.45cm, inner sep=0ex] 
	(z1) at (2,-2) {\small $y$};
}
\newcommand{\outerNodes}{
	\nodesVX
	\nodesWY
}
\newcommand{\outerNodesColV}{
	\nodesVXcolV
	\nodesWY
}
\newcommand{\outerNodesColX}{
	\nodesVXcolX
	\nodesWY
}
\newcommand{\outerNodesColVX}{
	\nodesVXcol
	\nodesWY
}
\newcommand{\nodesXY}{
	\node[draw,circle,minimum size=0.45cm, inner sep=0ex] 
	(b1) at (2  , 1) {\small $\overline{x}$};
	\node[draw,circle,minimum size=0.45cm, inner sep=0ex] 
	(d1) at (2  ,-1) {\small $\overline{y}$};
}
\newcommand{\nodesXYcol}{
	\node[draw,circle,minimum size=0.45cm, inner sep=0ex, nodeset two] 
	(b1) at (2  , 1) {\small $\overline{x}$};
	\node[draw,circle,minimum size=0.45cm, inner sep=0ex, nodeset one] 
	(d1) at (2  ,-1) {\small $\overline{y}$};
}
\newcommand{\nodesXYcolRev}{
	\node[draw,circle,minimum size=0.45cm, inner sep=0ex, nodeset one] 
	(b1) at (2  , 1) {\small $\overline{x}$};
	\node[draw,circle,minimum size=0.45cm, inner sep=0ex, nodeset two] 
	(d1) at (2  ,-1) {\small $\overline{y}$};
}
\newcommand{\nodesYcol}{
	\node[draw,circle,minimum size=0.45cm, inner sep=0ex, nodeset one] 
	(d1) at (2  ,-1) {\small $\overline{y}$};
}
\newcommand{\nodesYcolRev}{
	\node[draw,circle,minimum size=0.45cm, inner sep=0ex, nodeset two] 
	(d1) at (2  ,-1) {\small $\overline{y}$};
}

\newcommand{\nodesVWcolW}{
	\node[draw,circle,minimum size=0.45cm, inner sep=0ex, nodeset two] 
	(a1) at (1.5, 0) {\small $\overline{v}$};
	\node[draw,circle,minimum size=0.45cm, inner sep=0ex, nodeset one] 
	(c1) at (2.5, 0) {\small $\overline{w}$};
}

\newcommand{\innerNodesV}{
	\nodesXY
	\nodesVWcolW
}
\newcommand{\innerNodesVX}{
	\nodesXYcol
	\nodesVWcolW
}
\newcommand{\innerNodesVXmerged}{
	\nodesYcol
	\nodesVWcolW
}
\newcommand{\innerNodesWX}{
	\nodesXYcolRev
	\nodesVWcolW
}
\newcommand{\innerNodesWXmerged}{
	\nodesYcolRev
	\nodesVWcolW
}

\newcommand{\otherCrossingEdges}{
	\draw[-,gray,dotted, very thick] 
	(-0.25,1.5) to[bend left] (0.5,-2); 	%
	\draw[-,gray,dotted, very thick] 
	(3.25,2) to[bend left] (0.25,2); 		%
	\draw[-,gray,dotted, very thick] 
	(4,2) to[bend right] (3.75,-2.25); 	%
	\draw[-,gray,dotted, very thick] 
	(4,-1.5) to[bend left] (4.25,1.75); 	%
}
\newcommand{\otherEdges}{
	\draw[-,gray,dashed] (v) to (-0.5,1);
	\draw[-,gray,dashed] (v) to (-0.5,-1.5);
	\draw[-,gray,dashed] (v) to (-0.75,-0.5);
	\draw[-,gray,dashed] (v) to (-0.75,0.25);
	\draw[-,gray,dashed] (y) to (4.5,1.25);
	\draw[-,gray,dashed] (y) to (4.5,-1);
	\draw[-,gray,dashed] (y) to (4.75,0.5);
	\draw[-,gray,dashed] (z1) to (3,-2.5);
	\draw[-,gray,dashed] (z1) to (0.75,-2.5);
	\draw[-,gray,dashed] (w1) to (0.75,2.5);
}

\tikzset{snake it/.style={decorate, decoration={snake,amplitude=.75mm,segment length=3mm,post length=0mm}}}
\tikzset{cut/.style={snake it,red,line width=2pt,font={\bfseries}}}%
\tikzset{double/.style={blue,densely dotted, very thick}}
\tikzset{inCut/.style={black!80!white, decorate, decoration={snake,amplitude=.5mm,segment length=1.5mm,post length=0mm}}}
\tikzset{inCutGray/.style={gray, decorate, decoration={snake,amplitude=.5mm,segment length=1.5mm,post length=0mm}}}
\tikzset{inCutDouble/.style={blue,densely dotted, thick, decorate, decoration={snake,amplitude=.5mm,segment length=1.5mm,post length=0mm}}}
\tikzset{edge node/.style={rectangle,minimum height=0mm,inner sep=.5mm,align=center,fill=white}}

\tikzset{nodeset one/.style={fill=green!30}}
\tikzset{nodeset two/.style={fill=blue!30}}
\tikzset{colCutEdge/.style={blue!90!black}}
\tikzset{cut edge/.style={colCutEdge, ultra thick}}
\tikzset{colSkEdge/.style={red!90!black}}
\tikzset{skewness edge/.style={colSkEdge, ultra thick}}

\tikzset{example/.style={}}
\tikzset{example node/.style={circle,inner sep=0mm,minimum size=0.8cm,thick,draw,font={\huge}}}
\tikzset{example edge/.style={draw=black,very thick}}
\tikzset{crossing edge node/.style={pos=0.4}}

\tikzset{example edge node/.style={rectangle,minimum height=4mm,inner sep=.5mm,align=center,fill=white,scale=1.4}}

\usepackage{xparse}

\newcommand{\PFMC}{\textsc{PF-Max-Cut}\xspace}

\makeatletter
\newsavebox\myboxA
\newsavebox\myboxB
\newlength\mylenA
\newcommand*\xoverline[2][0.75]{%
    \sbox{\myboxA}{$\m@th#2$}%
    \setbox\myboxB\null%
    \ht\myboxB=\ht\myboxA%
    \dp\myboxB=\dp\myboxA%
    \wd\myboxB=#1\wd\myboxA%
    \sbox\myboxB{$\m@th\overline{\copy\myboxB}$}%
    \setlength\mylenA{\the\wd\myboxA}%
    \addtolength\mylenA{-\the\wd\myboxB}%
    \ifdim\wd\myboxB<\wd\myboxA%
       \rlap{\hskip 0.5\mylenA\usebox\myboxB}{\usebox\myboxA}%
    \else
        \hskip -0.5\mylenA\rlap{\usebox\myboxA}{\hskip 0.5\mylenA\usebox\myboxB}%
    \fi}
\makeatother
\newcommand{\N}[1]{\xoverline{#1}}
\newcommand{\pfmc}{\mathtt{MaxCut_{pf}}}

\newcommand{\MC}{\textsc{Max-Cut}\xspace}
\newcommand{\MCP}{\textsc{Max-Cut} problem\@\xspace}
\newcommand{\MCA}{\textsc{Max-Cut} algo\-rithm\@\xspace}

\newcommand{\nodeCut}{S}	%

\newcommand{\Cut}{\delta(\nodeCut)}	%

\newcommand{\crg}{\mathrm{cr}}
\newcommand{\mcr}{\mathrm{mcr}}

\newtheorem{thm}{Theorem}
\newtheorem{lem}[thm]{Lemma}
\newtheorem{defi}[thm]{Definition}
\newtheorem{cor}[thm]{Corollary}

\newcommand{\nover}[1]{\begingroup%
	#1\endgroup}
\newcommand{\changed}[1]{\begingroup%
	#1\endgroup}

\begin{document}
\doi{} %
\Issue{0}{0}{0}{0}{0} %
\HeadingAuthor{} %
\HeadingTitle{} %
\title{Maximum Cut Parameterized by \\ Crossing Number} %
\Ack{} %

\author[1]{Markus Chimani}{markus.chimani@uni-osnabrueck.de}%
\author[2]{Christine Dahn}{christine.dahn@cs.uni-bonn.de}%
\author[1]{Martina Juhnke-Kubitzke}{juhnke-kubitzke@uni-osnabrueck.de}
\author[3]{Nils M.~Kriege}{nils.kriege@cs.tu-dortmund.de}%
\author[2]{Petra Mutzel}{petra.mutzel@cs.uni-bonn.de}%
\author[1]{Alexander Nover}{anover@uni-osnabrueck.de}

\affiliation[1]{School of Mathematics/Computer Science, University Osnabrück, Germany\\
\texttt{\{markus.chimani,juhnke-kubitzke,anover\}@uni-osnabrueck.de}} %
\affiliation[2]{Institute for Computer Science, University of Bonn, Germany\\
\texttt{\{christine.dahn,petra.mutzel\}@cs.uni-bonn.de}} %
\affiliation[3]{Department of Computer Science, TU Dortmund University, Germany\\
\texttt{nils.kriege@cs.tu-dortmund.de}} %

\submitted{August 2019}%
\reviewed{}%
\revised{}%
\accepted{}%
\final{}%
\published{}%
\type{}%
\editor{}%

\maketitle

\setcounter{page}{1}

\begin{abstract}
Given an edge-weighted graph $G$ on $n$ nodes, 
the NP-hard \MCP asks for a node bipartition such that the sum of edge weights joining the different partitions is maximized.
We propose a fixed-parameter tractable algorithm parameterized by the number $k$ of crossings in a given drawing of $G$.
Our algorithm achieves a running time of $\mathcal{O}(2^{k} \cdot p(n+k))$, where $p$ is the polynomial running time for planar \MC.
The only previously known similar algorithm~\cite{DKM18} is restricted to embedded 1-planar graphs (i.e., at most one crossing per edge)
and its dependency on $k$ is of order $3^k$.
Finally, combining this with the fact that crossing number is fixed-parameter tractable with respect to\ itself, we see that \MC is fixed-parameter tractable with respect to\ the crossing number, even without a given drawing.
\nover{Moreover, the results naturally carry over to the minor-monotone-version of crossing number.}

 \end{abstract}

\bigskip

\section{Introduction}

Cut problems in graphs are a well-established class of problems attracting interest since the beginning of modern algorithmic research.
Given an edge-weighted undirected graph, the \MCP 
asks for a node partition into two sets, such that the sum of the weights of the edges between the partition sets is maximized.  
The problem is getting increasing attention in the literature due to its applicability to various scenarios:
these range from $\ell^1$-embeddability~\cite{app1}, to the layout of electronic circuits \cite{BarahonaGJR88,DDJMRR1995},
to solving Ising spin glass models, which are of high interest in physics~\cite{Barahona1982}.
Besides the theoretical merits, such models need to be solved in adiabatic quantum computation~\cite{McGeoch2014}.
Furthermore, quadratic unconstrained binary optimization (QUBO)
problems can be solved via \MC.
Many combinatorial optimization problems can be stated in the form of QUBO such as multicommodity-flow problems, maximum clique, vertex cover, scheduling,  and many others.
Also see~\cite{app2,app1} for a more in-depth overview on applications.

The \MCP has been shown to be NP-hard for general graphs~\cite{Karp72}.
Papadimitriou and Yannakakis~\cite{PapadimitriouY1991} have shown that the \MCP is even APX-hard, i.e., there does not exist a polynomial-time approximation scheme unless P=NP.
Goemans and Williamson proposed a randomized constant-factor approximation algorithm~\cite{GoemansW1995}, which has been derandomized by Mahajan and Ramesh~\cite{MahajanR1999}, 
achieving a ratio of~0.87856.
Several special cases of the problem allow polynomial algorithms:
If the weights of all edges are negative we obtain a \textsc{Min-Cut} problem, which 
can be solved, e.g., via network flow.
Other special cases are, e.g., graphs without long odd cycles~\cite{GrotschelN84} or weakly bipartite graphs~\cite{GroetschelPulleyblank1981}.
The case of planar input graphs is of particular interest.
Orlova and Dorfman~\cite{OrlovaD72} and Hadlock~\cite{Hadlock75} have shown how to solve \MC in polynomial time for unweighted planar graphs.
Those algorithms can be extended to weighted planar graphs;
the currently fastest algorithms for the weighted case have been suggested by Shih et al.~\cite{ShihWuKuo90} and by Liers and Pardella~\cite{LiersP12},
and achieve a running time of $p(n)=\mathcal{O}(n^{3/2}\log n)$ on planar graphs with $n$ nodes.
Barahona has shown that the planarity condition can be relaxed to graphs not contractible to $K_5$~\cite{Barahona83}.

\nover{
Similarly, it has been shown that \MC can be solved in polynomial time if the graph can be embedded on a surface of constant genus $g$~\cite{Galluccio1998O,GalluccioLV01}. 
However,  the edge-weights have to be restricted to integers whose absolute values are bounded by a polynomial in the input.
The central idea of this algorithm is to write the generating function of cuts as a linear combination of $4^g$ Pfaffians.
This algorithm is in fact highly non-trivial to realize: In order to obtain an implementable algorithm, \cite{GalluccioLV01} reports on a scheme, which depends on the existence of sufficiently many prime numbers within a given interval, which cannot be guaranteed in general.
}

A graph is \emph{1-planar} if it allows a drawing where each edge is involved in at most one crossing. A \emph{1-plane} graph is such a graph, together with an embedding realizing this property.
The \MCP on 1-plane graphs with $k$ edge crossings has recently been shown to be \emph{fixed-parameter tractable} (FPT) with parameter $k$~\cite{DKM18}.
More precisely, it was shown that such instances can be solved in $\mathcal{O}(3^k\cdot p(n))$ time, where $p(n)$ is the running time of a polynomial-time \MCA\ on planar graph with $n$ nodes, e.g., \cite{LiersP12, ShihWuKuo90}. There are no restrictions on the edge weights.

\paragraph*{Our contribution.}
In this paper, we improve on the latter result in several ways: 
Firstly, we drop the requirement of 1-planarity, i.e., we consider graphs that
can be drawn with at most $k$ crossings (even if multiple such crossings lie on the same edge).
We therefore handle the case of the well-established notion of the graph's \emph{crossing number}. 
Secondly, we reduce the runtime dependency on $k$ from $3^k$ to $2^k$.
Finally, unlike the previous result, our approach can be extended to an 
FPT algorithm which does not even require a crossing-realizing drawing as an input; 
however, this increases the %
running time and requires a deep algorithm from the literature as a black box \cite{Grohe04,KR07}.
Interestingly, we achieve these results %
by a \emph{simpler} approach (compared to \cite{DKM18}).
\nover{Comparing our algorithm with~\cite{GalluccioLV01}, we have no restrictions on the edge weights. 
Even in the restricted scenario, our algorithm is faster for graphs whose crossing number is at most twice its genus.
Furthermore, we require only easily-implementable data structures and subalgorithms (if we are given a crossing-realizing drawing), compared to advanced methods from algebra.}

The general idea of our algorithm is to recursively get rid of each crossing, each time resulting in two new subinstances.
We end up with a set of up to $2^k$ planar graphs, each of which can be solved using a known polynomial-time \MCA for planar graphs.
The maximum over all these subinstances then yields a maximum cut in the original instance.

\nover{Finally, we consider parameterizing the problem by the \emph{minor} crossing number (see below for details). This measure is always at most the graph's crossing number. While the exponential dependency on the respective parameter is identical, the running time only slightly increases in its polynomial part.}

\changed{
\paragraph*{Independent work.}
Kobayashi et al.\cite{KKMT19I} independently and simultaneously obtained 
another fixed-parameter tractable algorithm for \MC parameterized by crossing number achieving the same running time.
However, while we can always stay in the realm of maximum cuts when solving subinstances, they have to consider maximum weighted $b$-factor problems. 
Their preprint was uploaded to arXiv shortly after ours~\cite{CDJKMN19,KKMT19}.
}

\paragraph*{Organization of the paper.} Section \ref{se:preliminaries} recapitulates the basic definitions for cuts and crossings.
In Section~\ref{se:algorithm}, we present our new algorithm and prove its correctness and running time.
\nover{Section~\ref{se:minor} extends the results to the minor crossing number case.}
We end with a conclusion and open problems in Section~\ref{se:conclusion}.

\section{Preliminaries}
\label{se:preliminaries}

Throughout this paper we consider undirected edge-weighted graphs.
The input for our \MCP is 
a graph $G=(V,E,c)$, where $c_e\in \mathbb R$ denotes the (positive or negative) weight of edge $e$, for each edge $e\in E$. %
A partition of the nodes $V$ into two sets $S \subseteq V$ and $\overline{S}=V\setminus S$
defines the \emph{cut} $\Cut=\{uv \in E \mid u\in S \Leftrightarrow v\not\in S \}$.
The \emph{value} $c(\Cut) = \sum_{e\in \Cut} c_{e}$ of a cut is the sum of all edge weights of the edges in the cut.
Given $G$, the \MCP asks for a cut with highest value.
Since a graph can have multiple cuts of equal value, only the value of a maximum cut is unique, not the cut itself.

A \emph{non-degenerated drawing} of a graph in the plane is a map of its nodes to distinct points in $\mathbb{R}^2$, and a map of its edges
to curves connecting the respective endpoints, not including the points of any other node. Any point mapped to the plane 
either corresponds to a graph node, or is contained in at most two edge curves.
A shared non-endpoint 
between two curves is called a \emph{crossing}. 

A graph is \emph{planar} if it admits a drawing without any crossings.
It is well known that planarity can be tested in linear time \cite{HopcroftT73}.
For non-planar graphs it is natural to ask for a drawing with as few crossings as possible.
The smallest such number is the \emph{crossing number} $\crg(G)$ of $G$.
Not only is it NP-hard to compute $\crg(G)$~\cite{GareyJohnson83}, but even the so called \emph{realizability} problem
turns out to be NP-hard~\cite{Kratochvil91}: Given a graph $G$ and a set $X$ of edge pairs, is there a drawing $\mathcal{D}$ of $G$ such that
$X$ contains an edge pair if and only if the pair's two edge curves cross in $\mathcal{D}$?
The key problem in testing realizability is that it is hard to figure out whether there exist \emph{orderings}
of the crossings along the respective edges that allow the above properties.

Therefore, sometimes more restricted crossing variants are considered. For example, \emph{1-planar} graphs admit drawings
where every edge is involved in at most one crossing. 
Not all graphs can be drawn in such a way, since 1-planar graphs can have at most $4|V|-8$ edges; %
also, the 1-planar number of crossings is in general larger than $\crg(G)$~\cite{PT,CKMV19}.

For a general drawing (not necessarily 1-plane), we typically encode its crossings as a \emph{crossing configuration} $\mathcal{X}$.
Therein, we not only store the pairs of edges that cross, but for each edge also the order of the crossings as they occur along its curve.
\changed{The feasibility of a crossing configuration can be tested in time linear in $|V|+|\mathcal{X}|$ by replacing crossings with dummy nodes of degree 4, testing planarity, and checking the cyclic order around dummy nodes.}%
\footnote{\changed{In general, a specified edge pair may not really cross but merely ``touch''; this is trivial to detect after testing planarity by checking the cyclic order of the edges around the dummy node. 
Given such a ``flawed'' configuration, we trivially obtain one with less crossings by removing such crossing pairs from $\mathcal{X}$.%
}}
Although we will not require this fact in the following, this also allows us to efficiently deduce a drawing that
respects~$\mathcal{X}$.
It is well understood that we can restrict ourselves to \emph{good drawings} when considering the (traditional) crossing number
of graphs: adjacent edges never cross and no edge pair crosses \changed{more than once}.

\section{Algorithm}
\label{se:algorithm}
Our main idea for computing the maximum cut in an embedded weighted graph is to eliminate its crossings one by one. In the end, we use a \MCA for planar graphs.
We first introduce a slight variant of \MC: %
\begin{defi}[Partially-Fixed Maximum Cut, \PFMC]\ 
 Given an edge weighted graph $G=(V,E,c)$ and a set of \emph{fixed} edges $F\subseteq E$, find a cut of maximum value that contains all elements of $F$.
\end{defi}
A cut is \emph{feasible} if it contains $F$. A \PFMC instance is \emph{infeasible} if it does not allow a feasible cut. It is easy to
see that an instance is infeasible if and only if $F$ contains a cycle of odd length.
We denote a maximum objective value by $\pfmc(G,F)$, and let $\pfmc(G,F)=-\infty$ for infeasible instances.

\changed{Observe that (as for \MC) we do not need to consider a given crossing configuration $\mathcal{X}$ as part of the problem description (see Corollary~\ref{crFPTcor}). However, since having $\mathcal{X}$ allows for simplifications and a better running time, we will for now assume that we are given the graph together with a crossing configuration $\mathcal{X}$.\footnote{\changed{As noted above, we may assume that no such $\mathcal{X}$ ever specifies ``touching points''; we can reduce such configurations whenever our algorithm retrieves a new crossing configuration.}}
We will explain later how to remove this assumption in Corollary~\ref{crFPTcor}.
}

Given any edge $vw$ with weight $c_{vw}$ in a \PFMC instance, 
we define the operation to \emph{bisubdivide $vw$ at $v$} as follows:
Subdivide $vw$ twice, i.e., replace $vw$ by a path of length $3$ with two new degree-2 nodes. 
We denote the new node incident to $v$ or $w$ by $\N{v}$ or $\N{w}$, respectively. We consider the notation $\N{\ \cdot \ }$ an operand.\footnote{Observe that per recursion step, we will bisubdivide at most one edge per incident node (recall that adjacent edges never cross in good drawings). Thus, the above simple notation is unambiguous. In the graphs of the subproblems, see below, we may assume the nodes to be named afresh, and thus we may again perform bisubdivisions without creating notational ambiguity.}

The edges $v\N{v}$ and $\N{v}\N{w}$ have weight $0$, $\N{w}w$ retains the weight $c_{vw}$.
Furthermore, we add $v\N{v},\N{v}\N{w}$ to $F$, and if $vw\in F$, we replace it in $F$ by $\N{w}w$.
Clearly, both $v\N{v},\N{v}\N{w}$ will be in any feasible cut;
node $\N{w}$ will always lie in the same partition set as $v$, and $\N{v}$ in the other (cf. Figure \ref{fig:G'}). Most importantly this gives:

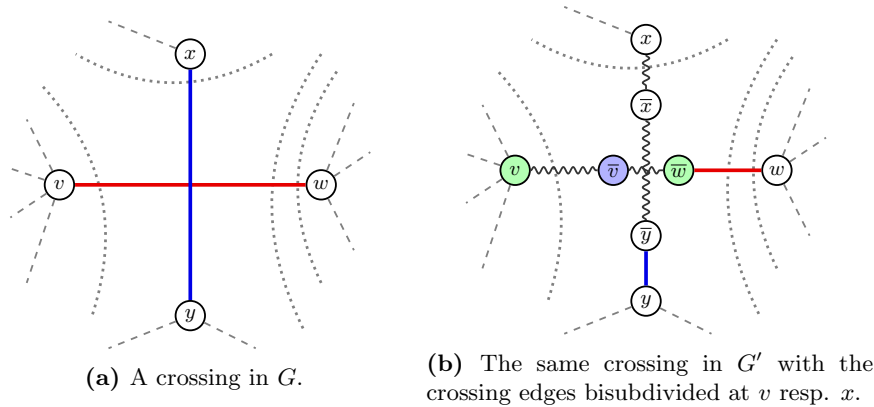
\begin{figure}[tb]
	\centering
	\begin{subfigure}{.49\linewidth}
		\centering		
		\scalebox{0.87}{%
\begin{tikzpicture}
\begin{scope} [thick] 
	\outerNodes		
	\otherCrossingEdges
	\otherEdges
		
	\draw[skewness edge]  (v) -- (y); 
	\draw[cut edge] (w1) -- (z1);
\end{scope}
\end{tikzpicture} %
} 
		\caption{A crossing in $G$.}\label{fig:G}
	\end{subfigure}
	\begin{subfigure}{.49\linewidth}
		\centering		
		\scalebox{0.87}{%
\begin{tikzpicture} 
\begin{scope} [thick] 
	\outerNodesColV	
	\otherCrossingEdges	
	\innerNodesV
	\otherEdges
	
	\draw[inCut] (v) -- (a1);
	\draw[inCut] (a1) -- (c1);
	\draw[skewness edge]  (c1) -- (y); 
	
	\draw[inCut] (w1) -- (b1);
	\draw[inCut] (b1) -- (d1);
	\draw[cut edge] (d1) -- (z1);
\end{scope}
\end{tikzpicture} %
}
		\caption{The same crossing in $G'$ with the crossing edges bisubdivided at $v$ resp. $x$.}\label{fig:G'}
	\end{subfigure}
	\caption{The situation at a crossing between $vw$ and $xy$ in $G$. In $G'$, the two edges of the crossing are bisubdivided at $v$ and $x$, respectively, and the zig-zag edges are added to the set of fixed edges~$F'$.
		As an example, the node coloring at $v,\N{v},w$ gives a partition of these nodes that is forced by the respective newly added edges in $F'$. (Dashed and dotted edges show examples of other edges in $G$, resp.\ $G'$.)}
	\label{fig:construction}
\end{figure}

\begin{lem}\label{bisubinvar}
The feasible cuts in an original \PFMC instance $\langle G,F\rangle$ are in 1-to-1 correspondence to feasible cuts of equal value in a bisubdivided instance $\langle G',F'\rangle$.
\end{lem}
\begin{proof}
Let $vw$ be the edge in $G$ that is bisubdivided at $v$ to obtain $\langle G',F'\rangle$.
By construction, we know that both edges $v\N{v}$, $\N{v}\N{w}$ have cost $0$ and are in $F'$, and thus in any $F'$-feasible cut. Consequently, in any $F'$-feasible cut, $v$ and $\N{w}$ will lie in a 
common partition set. 
Let $S'\subset V(G')$ be a node subset that induces some feasible (with respect to $F'$) cut in $G'$.
Then, the node set $S=S'\setminus\{\N{v},\N{w}\}$ induces a feasible (with respect to\ $F$) cut in $G$. 
Cut $\delta(S)$ contains edge $vw$ if and only if $\delta(S')$ contains $\N{w}w$. Since both these edges have identical cost, the total costs of both cuts are equal.

Inversely, let $S\subset V(G)$ be a node subset that induces some feasible (with respect to $F$) cut in $G$. 
Then, consider the cut in $G'$ induced
by $S'=S\cup\{s\}$, where $s=\N{w}$ if $v\in S$, and $s=\N{v}$ otherwise. Both fixed edges $v\N{v}$, $\N{v}\N{w}$ are in $\delta(S')$ and the cut is thus feasible. Again, 
 $\delta(S')$ contains edge $\N{w}w$ if and only if $\delta(S)$ contains $vw$, and both cut values are thus equal.
\end{proof}

When we \emph{identify} two nodes $a,b$ in a graph with one another, they become a common entity that is incident to all of their former neighbors.
We will only identify nodes that are neither adjacent nor share neighbors. 
When identifying nodes in $G$ of some \PFMC instance $\langle G,F\rangle$,
the set $F$ is retained, subject to replacing the edges formerly incident to $a$ or $b$ with their new counterparts.

\medskip
We are now ready to describe our algorithm.
We are given a \MC instance $G=(V,E,c)$, together with some crossing configuration $\mathcal{X}$ with $k$ crossings.
Let $F=\emptyset$ be the set of fixed edges and consider $\langle G,F\rangle$ as a \PFMC instance.
From $\langle G,F,\mathcal{X}\rangle$, we pick a crossed edge $vw$, and derive two new triplets $T_i=\langle G_i,F_i,\mathcal{X}_i\rangle$, for $i\in\{v,w\}$.
Both derived crossing configurations $\mathcal{X}_i$ attain at most $k-1$ crossings and we can call our algorithm recursively on $T_v$ and $T_w$.
As a base case, the derived graphs become planar and (after a preprocessing to deal with the fixed edges) we apply an efficient \MCA for planar graphs.
The solutions of $\langle G_i,F_i\rangle$, for $i\in\{v,w\}$, yield a solution of $\langle G,F\rangle$.
Observe, however, that $\langle G_i,F_i\rangle$ may become infeasible.\\

Let us describe this recursion step formally (cf.\ also Figures~\ref{fig:construction} and~\ref{fig:construction2}).
We define the \emph{crossing split} operation that, given a triplet $\langle G,F,\mathcal{X}\rangle$, yields the two triplets $T_v$ and $T_w$:
Let $\langle G=(V,E,c),F\rangle$ be a \PFMC instance and $\mathcal{X}$ a crossing configuration of $G$.
Consider a crossing $\chi\in\mathcal{X}$ with crossing edges $vw$ and $xy$.
For $j\in\{v,w,x,y\}$, let $Y_j$ be the \textit{ordered sets of crossings} in $\mathcal{X}$ between $j$ and $\chi$ (cf.\ the dotted edges in Figure \ref{fig:construction}: 
\changed{e.g., the crossings between the two dotted edges and $\N{w}w$ in Figure \ref{fig:G'} are in $Y_w$ as they are between $\chi$ and $w$ in Figure \ref{fig:G}).}
Let the intermediate instance $\langle G',F'\rangle$ be obtained from $\langle G,F\rangle$ by bisubdividing $vw$ at $v$ and bisubdividing $xy$ at~$x$.
For $i\in\{v,w\}$, let $\langle G_i,F_i\rangle$ be the \PFMC instance obtained from $\langle G',F'\rangle$ by identifying $\N{x}$ with $\N{i}$ \changed{(see Figures \ref{fig:G'_x_merged} and \ref{fig:G'_y_merged})}.
\changed{Intuitively, the two graphs obtained by the identifications represent the two possibilities whether $x$ is on the same side of the cut as $v$ or not.}
We obtain a corresponding crossing configuration $\mathcal{X}_i$ from $\mathcal{X}$ by removing $\chi$ and placing the crossings $Y_j$ (retaining their order) on the edge $j\N{j}$, for all $j\in\{v,w,x,y\}$. 
The triplets $T_v=\langle G_v,F_v,\mathcal{X}_v\rangle$ and $T_w=\langle G_w,F_w,\mathcal{X}_w\rangle$ are the results of the crossing split operation with respect to\ $\langle \chi,vw,xy\rangle$.

\begin{lem}\label{lem:recurse}
 Let $\langle G=(V,E,c),F\rangle$ be a \PFMC instance and $\mathcal{X}$ a crossing configuration of $G$ with $k$ crossings.
 Let $\chi\in\mathcal{X}$ be any crossing with some crossing edges $vw$ and $xy$, and consider the crossing split operation with respect to\ $\langle \chi,vw,xy\rangle$.
 For $i\in\{v,w\}$, let $\langle G_i,F_i,\mathcal{X}_i\rangle$ be the resulting triplets.
 Then we have:
 \begin{enumerate}
  \item for $i\in\{v,w\}$, $\mathcal{X}_i$ is a feasible crossing configuration for $G_i$ with at most $k-1$ crossings; and
  \item $\pfmc(G,F) = \max_{i\in\{v,w\}} \{\ \pfmc(G_i,F_i)\ \}$.
 \end{enumerate}
\end{lem}

\begin{proof}
Consider any drawing $\mathcal{D}$ of $G$ realizing $\mathcal{X}$.
By routing the new paths ($v\N{v},\N{v}\N{w},\N{w}w$ resp. $x\N{x},\N{x}\N{y},\N{y}y$) along the curves of their original edges ($vw$ resp. $xy$) 
we obtain a drawing $\mathcal{D}'$ of $G'$ from $\mathcal{D}$.
Thereby, for  $j\in\{v,w,x,y\}$, we place the new nodes $\N{j}$ in a close neighborhood of $\chi$ on the curve segment between $j$ and $\chi$, so that $\N{x}\N{y}$ is only crossed by $\N{v}\N{w}$ and vice versa. Note that
the number of crossings in $\mathcal{D}'$ is equal to that of $\mathcal{D}$,
since all crossings in $Y_j$ in $\mathcal{D}$ are transferred to the edge $j\N{j}$ in $\mathcal{D}'$, for all $j\in \{v,w,x,y\}$, and 
the original crossing $\chi$ between $xy$ and $vw$ in $\mathcal{D}$ has a counterpart $\chi'$ in $\mathcal{D}'$ between the edges $\N{v}\N{w}$ and $\N{x}\N{y}$.
Since the edges $\N{v}\N{w}$ and $\N{x}\N{y}$ are crossing free except for $\chi'$, we can follow (in a close neighborhood) the curves of $\N{v}\N{w}$ from any of its end points up to
$\chi'$, and onwards from there along the curve of $\N{v}\N{w}$ to any of its end points. Since these routes are crossings-free, we call them \emph{free routes}.
When we now identify $\N{x}$ with $\N{v}$, we can locally redraw
our drawing such that $\chi'$ vanishes and no other crossings arise, see Figure~\ref{fig:G'_x_merged}. 
Observe that $\N{x}$ has precisely two neighbors: $\N{y}$ and $x$. The identification is thus such that we may remove $\N{x}$ and insert edges $\N{y}\N{v}$ and $x\N{v}$ instead.
The former can trivially be drawn without any crossings along the free route between $\N{y}$ and $\N{v}$. The curve for the latter edge is the concatenation of the former curve of $x\N{x}$ and the free
route between $\N{x}$ and $\N{v}$. The number of crossings along the edge $x\N{x}$ (with now $\N{x}=\N{v}$) does thus not change.
We can perform the analogous redrawing when 
identifying $\N{x}$ with $\N{w}$, see Figure~\ref{fig:G'_y_merged}. This establishes claim~(1).

Two nodes $v$ and $x$ can either be on the same side of a cut, or they are on opposite sides. Therefore, 
we create two new subproblems in which $v$ and $x$ are in the same partition set or not, respectively.
In $G_v$ (where we identify $\N{x}$ with $\N{v}$), we have a path of two edges between $v$ and $x$ (namely $v\N{v}$ and $\N{v}x$), both of which are in $F_v$. Thus, $v$ and $x$ have to be in the same partition set, see Figure~\ref{fig:G'_x_merged}.
Conversely, in $G_w$ (where we identify $\N{x}$ with $\N{w}$), we have a path of three edges between $v$ and $x$ (namely $v\N{v}$, $\N{v}\N{w}$, and $\N{w}x$), all of which are in $F_w$.
Thus, $v$ and $x$ have to be in different partition sets, see Figure~\ref{fig:G'_y_merged}.
We can see that the respective constructions do not induce any further restrictions on the set of cuts. In particular, both derived instances
still allow any partition choice between $w$ and $x$, between $w$ and $y$, and between $x$ and $y$. Overall, every feasible cut in $\langle G',F'\rangle$ can be
realized either in $\langle G_v,F_v\rangle$ or in $\langle G_w,F_w\rangle$.

If we know the maximum cut in instance $\langle G_v,F_v\rangle$ and the maximum cut in instance $\langle G_w,F_w\rangle$, we can pick the larger of these two cuts and transfer it back to $\langle G',F'\rangle$.
By applying Lemma~\ref{bisubinvar} twice (once for the bisubdivision of $vw$ at $v$ and once for the bisubdivision of $xy$ at~$x$), the maximum cut in $\langle G',F'\rangle$ induces a maximum cut in $\langle G,F\rangle$ of the same value.
Claim (2) follows.
\end{proof}

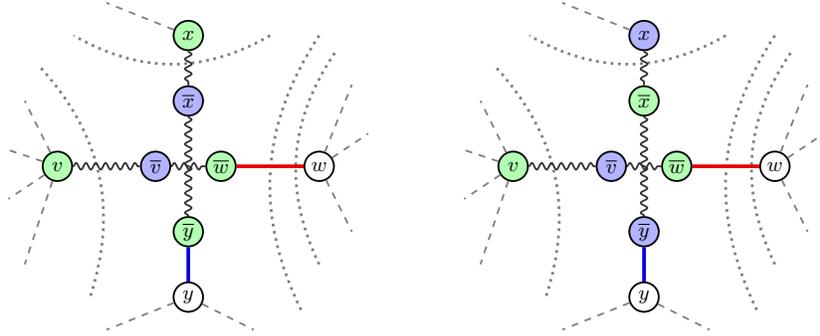
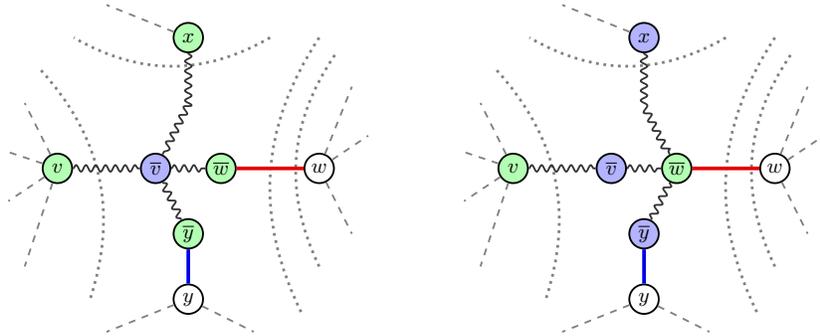
\begin{figure}[tb]
	\centering
	\begin{subfigure}{.49\linewidth}
		\centering		
		\scalebox{0.87}{%
\begin{tikzpicture}
\begin{scope} [thick] 
	\outerNodesColVX
	\otherCrossingEdges
	\innerNodesVX
	\otherEdges
	
	\draw[inCut] (v) -- (a1);
	\draw[inCut] (a1) -- (c1);
	\draw[skewness edge]  (c1) -- (y); 
	
	\draw[inCut] (w1) -- (b1);
	\draw[inCut] (b1) -- (d1);
	\draw[cut edge] (d1) -- (z1);

\end{scope}
\end{tikzpicture} %
}
		\caption{Induced partition in $G'$ with $x$ and $v$ on the same side of the partition.}\label{fig:G'_x}
	\end{subfigure}
	\begin{subfigure}{.49\linewidth}
		\centering		
		\scalebox{0.87}{%
\begin{tikzpicture}
\begin{scope} [thick] 
	\outerNodesColX
	\otherCrossingEdges
	\innerNodesWX
	\otherEdges
	
	\draw[inCut] (v) -- (a1);
	\draw[inCut] (a1) -- (c1);
	\draw[skewness edge]  (c1) -- (y); 
	
	\draw[inCut] (w1) -- (b1);
	\draw[inCut] (b1) -- (d1);
	\draw[cut edge] (d1) -- (z1);
\end{scope}
\end{tikzpicture} %
}
		\caption{Induced partition in $G'$ with $x$ and $v$ on different sides of the partition.}\label{fig:G'_y}
	\end{subfigure}
	
	\vspace*{12pt}
	\begin{subfigure}{.49\linewidth}
		\centering		
		\scalebox{0.87}{%
\begin{tikzpicture}
\begin{scope} [thick] 
	\outerNodesColVX
	\otherCrossingEdges
	\innerNodesVXmerged
	\otherEdges
	
	\draw[inCut] (v) -- (a1);
	\draw[inCut] (a1) -- (c1);
	\draw[skewness edge]  (c1) -- (y); 
	
	\draw[inCut] (w1) -- (2  , 1) -- (a1);
	\draw[inCut] (a1) -- (d1);
	\draw[cut edge] (d1) -- (z1);

\end{scope}
\end{tikzpicture} %
}
		\caption{In $G_{v}$, $\N{x}$ is identified with $\N{v}$.}\label{fig:G'_x_merged}
	\end{subfigure}
	\begin{subfigure}{.49\linewidth}
		\centering		
		\scalebox{0.87}{%
\begin{tikzpicture}
\begin{scope} [thick] 
	\outerNodesColX
	\otherCrossingEdges
	\innerNodesWXmerged
	\otherEdges
	
	\draw[inCut] (v) -- (a1);
	\draw[inCut] (a1) -- (c1);
	\draw[skewness edge]  (c1) -- (y); 
	
	\draw[inCut] (w1) -- (2  , 1) -- (c1);
	\draw[inCut] (c1) -- (d1);
	\draw[cut edge] (d1) -- (z1);
\end{scope}
\end{tikzpicture} %
}
		\caption{In $G_{w}$, $\N{x}$ is identified with $\N{w}$.}\label{fig:G'_y_merged}
	\end{subfigure}
	
	\caption{An illustration of the two cases where $v$ and $x$ are either on the same side of the partition (a/c) or on opposite sides (b/d). In the two graphs $G_{v}$ and $G_{w}$, the crossing was removed while retaining the partition property. The node coloring gives a partition of the nodes that is induced by the newly added edges in $F'$, resp.\ $F_v$ or $F_w$. (Dashed and dotted edges show examples of other edges in $G'$, resp.\ $G_v$ or $G_w$.)} 
	\label{fig:construction2}
\end{figure}

If we are in a base case -- the considered graph is planar -- we can use an efficient \MCA for planar graphs:
\begin{lem}\label{lem:basecase}
Consider a \PFMC instance $\langle G=(V,E,c),F\rangle$ with a planar graph $G$. 
Let $p(|V|)$ be a polynomial upper bound on the running time of a \MCA on the planar graph $G$.
We can compute an optimal solution to $\langle G,F\rangle$ -- or decide that the instance is infeasible -- in $\mathcal{O}(p(|V|))$.	
\end{lem}
\begin{proof}
We transform the \PFMC instance into a traditional \MC instance by attaching a large weight to the edges in $F$.
Namely, we add $M$ to the weight of each edge $f\in F$, where $M=2\cdot\sum_{e\in E} |c_{e}|$. 
The omission of a single edge of $F$ from the solution cut (even if picking all other edges of positive weight) will already result in a worse objective value than picking all of $F$ and all edges of negative weight.
The instance is infeasible if and only if the computed cut does not contain all of $F$; this can also be deduced purely by checking whether the objective value is at least $M\cdot |F|+\sum_{e\in E : c_e<0} c_{e}$.
\end{proof}

We proved our lemma above for a general case (by adding $M$ to the weight of each edge in $F$), but in fact we only require a slightly weaker version, since in our algorithm $c_f=0$ for all $f\in F$. Thus it suffices to set $c_f = M$ instead of adding $M$ to $c_f$.
Using any of the currently fastest \MC algorithms for planar graphs~\cite{LiersP12,ShihWuKuo90} leads to $\mathcal{O}(|V|^{3/2}\log |V|)$ time in the above lemma.
We could speed-up infeasibility detection by checking whether $F$ contains a cycle of odd length prior to the transformation; while this only requires
$\mathcal{O}(|V|)$ time via depth-first search, the overall asymptotic runtime for the lemma's claim does of course not improve.

\begin{thm}\label{crFPT}
Let $G=(V,E,c)$ be an edge-weighted graph and $\mathcal{X}$ a crossing configuration of $G$ with $k$ crossings.
Let $p(n)$ be a polynomial upper bound on the running time of a \MCA on planar graphs with $n$ nodes.
We can compute a maximum cut in $G$ in $\mathcal{O}(2^k \cdot p(|V|+k))$ 
time.	
\end{thm}
\begin{proof}
As described above, we solve the instance by considering the \PFMC instance $\langle G,F=\emptyset\rangle$ together with $\mathcal{X}$.
Thus the triplet $\langle G,F,\mathcal{X}\rangle$ forms the initial input of our recursive algorithm $\mathcal{R}$.

Algorithm $\mathcal{R}$ proceeds as follows on a given triplet:
If the triplet's graph is planar, we solve $\langle G,F\rangle$ via Lemma~\ref{lem:basecase}.
Otherwise, we use Lemma~\ref{lem:recurse} to obtain two new input triples $T_v, T_w$, for each of which we call $\mathcal{R}$ recursively. 
Their returned solutions (i.p., 
their solution values) induce the optimum solution for the current input triplet. 
However, while the number of crossings decreases by (at 
least) one per recursion step, the graph's size increases by three nodes.

The runtime complexity follows from the fact that we consider two choices per crossing in the given $\mathcal{X}$, and thus construct
$2^k$ graphs. For each such graph, which has $|V|+3k$ nodes, we run the planar \MCA.
\end{proof}

Above, we trivially have $k\in\mathcal{O}(|V|^4)$ and thus $|V|+k\in\mathcal{O}(\mathrm{poly}(|V|))$.
\begin{cor}
The above algorithm is an FPT algorithm with parameter $k$, provided that a crossing configuration $\mathcal{X}$ with $k$ crossings is part of the input.
Moreover, the attained running time is polynomial for any $k\in\mathcal{O}(\log |V|)$.
Using the currently fastest \MCA for planar graphs \cite{LiersP12,ShihWuKuo90}, our algorithm yields a running time of $\mathcal{O}(2^k\cdot (|V|+k)^{3/2}\log (|V|+k))$.
\end{cor}

Quite sophisticated results by Grohe~\cite{Grohe04} and Karabayashi and Reed~\cite{KR07} show that the problem to compute the crossing number of a graph %
is in FPT (even in linear time) with respect to\ its natural parameterization:
Given a graph $G$ and a number $k\in\mathbb{N}$, we can answer the question ``$\crg(G)\leq k$ ?'' in time $\mathcal{O}(f(k)\cdot n)$.
In case of a yes-instance, we obtain a corresponding crossing configuration $\mathcal{X}$ as a witness. The computable function $f(k)$ is purely dependent on~$k$. %
However, the dependency $f(k)$ is double exponential, and the algorithm far from being practical. Still, these results 
formally allow us to get rid of the requirement that $\mathcal{X}$ is part of the input:
\begin{cor}
\label{crFPTcor}
Given an edge-weighted undirected graph $G$. Computing a maximum cut in $G$ is FPT with parameter~$\crg(G)$.
\end{cor} %

\section{Minor Crossing Number}
\label{se:minor}

We say $G$ is a \emph{minor} of $H$, denoted by $G \preceq H$,  if $G$ can be obtained from $H$ by deletion and contraction of edges.
The \emph{minor crossing number} of $G$ is given by $\mcr(G)=\min\{\crg(H) \mid G \preceq H \}$.
A \emph{realization} of $\mcr(G)$ is a pair $(H, \mathcal{X})$ with $G \preceq H$ and $\mathcal{X}$ being a crossing configuration of $H$ with $\mcr(G)$ crossings. 
It is easy to see that for graphs $G'$ of maximum degree $3$ we have $\crg(G')=\mcr(G')$. Similarly, any graph $G$ allows a realizing graph $H$ ($\crg(H)=\mcr(G)$) of maximum degree $3$ where vertices of $G$ are replaced by disjoint cubic trees.

By definition we always have $\mcr(G) \leq \crg(G)$; as such $\mcr(G)$ can be a stronger FPT-parameter.
Also, in contrast to crossing number, the minor crossing number is monotone with respect to graph minors, i.e., the family $\{ G \mid \mcr(G) \leq k\}$ is minor closed. Thus, by \cite{ROBERTSON1995}, we can (theoretically) check whether $\mcr(G) \leq k$ in $\mathcal{O}(|V(G)|^3)$ time for fixed $k \in \mathbb{N}$.

Given a connected graph $G$ with $\mcr(G)=k$, we can obtain a graph $H$ from $G$ realizing $\mcr(G)$ in polynomial time as follows:
Choose a node $v$ of degree at least $4$.
Try different pairs of neighbors $w_1,w_2\in N(v)$ until finding the first with  $\mcr(\tilde{G}) \leq k$, where $\tilde{G}$ is obtained from $G$ by splitting $v$ into two nodes $v_1$ and $v_2$ with $N(v_1)=\{v_2,w_1,w_2\}$, $N(v_2)=(N(v)\cup\{v_1\}) \setminus \{w_1,w_2\}$\footnote{Observe that in general this splitting operation may increase $\mcr$; we search for a split (which has to exists) for which it does not increase. Since the split is an inverse minor operation, $\mcr$ can never decrease.}. We call the edge $v_1v_2$ a \emph{split edge}.
Iterating this for each high degree node, yields a graph $H$ of maximum degree $3$ realizing $\mcr(G)=\crg(H)$.
Note that $H$ has at most $\mathcal{O}(|E(G)|)$ nodes.

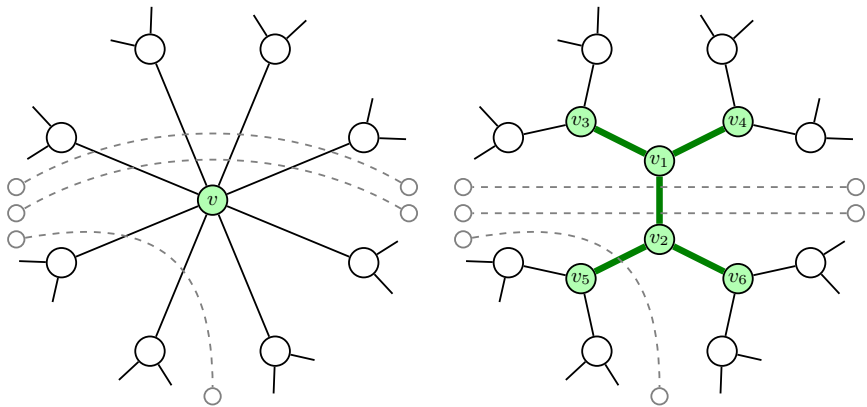
\begin{figure}[tb]
\begin{subfigure}{.49\linewidth}
\centering		
\scalebox{0.87}{%

\begin{tikzpicture}[every node/.style={draw,circle,minimum size=0.45cm, inner sep=0ex}]
\begin{scope} [thick] 
\node[nodeset one] (v) at (0,0) {\small $v$};
\foreach \phi in {1,2,...,8}{
	\node (\phi) at (\phi *360/8+360/16:2.5cm) {};
	\draw (\phi) -- (v);
	\draw (\phi) -- (\phi *360/8+360/16+10:2.9cm) {};
	\draw (\phi) -- (\phi *360/8+360/16-5:3.1cm) {};
	}
	
\node[draw,circle,minimum size=0.25cm, inner sep=0ex,gray] (g1) at (3, .2){};
\node[draw,circle,minimum size=0.25cm, inner sep=0ex,gray] (g2) at (-3,.2){};
\draw (g1) edge[dashed,gray ,in=30, out=150, distance=20mm](g2);

\node[draw,circle,minimum size=0.25cm, inner sep=0ex,gray] (g3) at (3, -.2){};
\node[draw,circle,minimum size=0.25cm, inner sep=0ex,gray] (g4) at (-3,-.2){};
\draw (g3) edge[dashed,gray ,in=30, out=150, distance=20mm](g4);

\node[draw,circle,minimum size=0.25cm, inner sep=0ex,gray] (g5) at (-3, -.6){};
\node[draw,circle,minimum size=0.25cm, inner sep=0ex,gray] (g6) at (0, -3){};
\draw (g5) edge[dashed,gray ,in=90, out=10, distance=20mm](g6);
\end{scope}
\end{tikzpicture} %
}%

\label{fig:minor_G}
\end{subfigure}%
\begin{subfigure}{.49\linewidth}
\centering		
\scalebox{0.87}{%

\begin{tikzpicture}[every node/.style={draw,circle,minimum size=0.45cm, inner sep=0ex}]
\begin{scope} [thick] 
\node[nodeset one] (v1) at (0,.6) {\small $v_1$};
\node[nodeset one] (v2) at (0,-.6) {\small $v_2$};
\node[nodeset one] (v3) at (-1.2,1.2) {\small $v_3$};
\node[nodeset one] (v4) at (1.2,1.2) {\small $v_4$};
\node[nodeset one] (v5) at (-1.2,-1.2) {\small $v_5$};
\node[nodeset one] (v6) at (1.2,-1.2) {\small $v_6$};

\foreach \phi in {1,2,...,8}{
	\node (\phi) at (\phi *360/8+360/16:2.5cm) {};
	\draw (\phi) -- (\phi *360/8+360/16+10:2.9cm) {};
	\draw (\phi) -- (\phi *360/8+360/16-5:3.1cm) {};
	}
	
\draw[line width = 1mm,green!50!black] (v1)--(v2);
\draw[line width = 1mm,green!50!black] (v1)--(v3);
\draw[line width = 1mm,green!50!black] (v1)--(v4);
\draw[line width = 1mm,green!50!black] (v2)--(v5);
\draw[line width = 1mm,green!50!black] (v2)--(v6);

\draw (v3)--(2);
\draw (v3)--(3);
\draw (v4)--(1);
\draw (v4)--(8);
\draw (v5)--(4);
\draw (v5)--(5);
\draw (v6)--(6);
\draw (v6)--(7);

\node[draw,circle,minimum size=0.25cm, inner sep=0ex,gray] (g1) at (3, .2){};
\node[draw,circle,minimum size=0.25cm, inner sep=0ex,gray] (g2) at (-3,.2){};
\draw (g1) edge[dashed,gray ](g2);

\node[draw,circle,minimum size=0.25cm, inner sep=0ex,gray] (g3) at (3, -.2){};
\node[draw,circle,minimum size=0.25cm, inner sep=0ex,gray] (g4) at (-3,-.2){};
\draw (g3) edge[dashed,gray ](g4);

\node[draw,circle,minimum size=0.25cm, inner sep=0ex,gray] (g5) at (-3, -.6){};
\node[draw,circle,minimum size=0.25cm, inner sep=0ex,gray] (g6) at (0, -3){};
\draw (g5) edge[dashed,gray ,in=90, out=10, distance=20mm](g6);
\end{scope}
\end{tikzpicture}  %
}%
\label{fig:minor_H'}
\end{subfigure}%
\caption{Visualization of the split operation to obtain an $\mcr$-realization. Left: part of a graph $G$ with $\crg(G)>\mcr(G)$. Right: part of $\tilde G$ after splitting $v$ five times. Bold green lines denote \emph{split edges}. 
	}
	\label{fig:minor_construction}
\end{figure}

Let $N=-3\cdot\sum_{e\in E(G)} |c_{e}|$. Attaching the weight $N$ to each split edge, we can make sure that these edges are not in any maximum cut of $H$.
Clearly, the cuts in $H$ not containing any split edges are in one-to-one correspondence with cuts in $G$.  
Using \Cref{crFPT}, we obtain an algorithm computing a maximum cut on $G$ parameterized by the $\mcr(G)$. Similarly to Corollary \ref{crFPTcor} we do not require an explicit realization as part of the input (using the above construction method for $H$).

\begin{cor}
(i) Let $G=(V,E,c)$ be an edge-weighted undirected graph with $\mcr(G)=k$, $(H,\mathcal{X})$ a realization of $\mcr(G)$, and $p(n)$ be a polynomial upper bound on the running time of a \MCA on planar graphs with $n$ nodes. 
We can compute a maximum cut in $G$ in $\mathcal{O}(2^k \cdot p(|E(G)|+k))$ time.	

(ii) Given an edge-weighted undirected graph $G$, computing a maximum cut in $G$ is FPT with parameter~$\mcr(G)$.
\end{cor}

\section{Conclusion and Open Problems}
\label{se:conclusion}
Given a graph together with a feasible crossing configuration with $k$ crossings, we previously only knew
that \MC is polynomial time solvable if $k$ is constant \emph{and} the graph is 1-planar, 
i.e., each edge is involved in at most one crossing. The runtime dependency on $k$ has been to the order of $3^k$~\cite{DKM18}.

Herein, we improved on this in several ways: Firstly, we decreased the dependency on $k$ to the order of $2^k$.
Secondly, we extended the applicability to \emph{any} graph with (at most) $k$ crossings: our parameter
becomes the true crossing number of the graph, without any 1-planarity restriction. 
This shows that \MC is in FPT with respect to\ the graph's crossing number. 
Moreover, we achieve these improvements by introducing simpler ideas than those proposed for the former result, yielding
an overall surprisingly simple algorithm.
\changed{Compared to the result of Kobayashi et al.\cite{KKMT19I}, we are able to stay within the realm of \MC. Finally, our result naturally carries over to the minor crossing number.}
\medskip

The \emph{skewness} of a graph is the minimum number of edges to remove such that the graph becomes planar.
The \emph{genus} of a graph is the minimum oriented genus of a surface onto which the graph can be embedded without crossings.
In FPT research, there are many algorithmic approaches that consider graphs with bounded genus $g$, see, 
e.g.~\cite{BodlaenderFLPST16,ChenKPSX07,EllisFF04,FominLRS11}.
However, the obtained FPT algorithms are typically parameterized by the objective value $z$, or by the combined parameter $(z,g)$.
There are much fewer results that obtain FPT algorithms parameterized purely with~$g$. Notable examples are the graph genus problem 
itself~\cite{MoharGenus} (where $z$ and $g$ coincide by definition),
and the graph isomorphism problem~\cite{IsoGenusKawa} (which generalizes the linear-time algorithm for the problem on planar graphs).
There are even fewer parameterized results with respect to\ skewness; 
the probably best known example is that maximum flow can be solved in the running time of planar graphs, if the graph's skewness is fixed~\cite{HochsteinWeihe07}.
Our above algorithm seems to be the first time that the crossing number has been proposed as an efficient non-trivial 
FPT parameter for any widely known problem.

Besides the weight-restricted case of~\cite{GalluccioLV01} (briefly described in the introduction), it is unclear whether \MC could be FPT
with respect to\ either skewness or genus. We deem this an interesting question for further research.

\bibliography{literatur}
\bibliographystyle{abbrv}

\end{document}